\documentclass[a4paper,11pt]{article}

\usepackage{lmodern}
\usepackage[T1]{fontenc}
\usepackage{microtype}
\usepackage[margin=2.5cm]{geometry}

\usepackage{amsmath}
\usepackage{amsfonts}
\usepackage{amssymb}
\usepackage{amsthm}

\usepackage[colorlinks=true,allcolors=blue]{hyperref}

\theoremstyle{plain}
\newtheorem{proposition}{Proposition}
\theoremstyle{remark}
\newtheorem*{comment}{Comment}
\newtheorem*{hypothesis}{Hypothesis}

\newcommand{\SMs}[1]{appendix~\ref{app:#1}}
\newcommand{\SM}[1]{\SMs{#1}}
\newcommand{\maeq}[1]{\eqref{#1}}
\newcommand{\maprop}[1]{Proposition~\ref{#1}}

\title{Symmetry-assisted computation of the magnetic field\\
at the site of a point dipole}
\author{D.~Pescia\\
\small Department of Physics, ETH Zurich, 8093 Zurich, Switzerland\\
\small \texttt{pescia@solid.phys.ethz.ch}\\
\small ORCID: 0000-0001-7436-418X}
\date{}

\begin{document}
\maketitle

\begin{abstract}
\noindent
In this paper, we revisit this concept of the magnetic field at the site of a magnetic dipole using symmetry arguments that are mathematically rigorous yet intuitive enough to be well-suited for upper-level undergraduate courses in electrodynamics and quantum mechanics. Furthermore, we present symmetry-based results that go beyond those discussed in the standard literature, which are directly relevant to fundamental problems in two-dimensional physics.

\end{abstract}

\section{Introduction}
The physical origin of magnetism in matter was a subject of long-standing debate in the nineteenth century \cite{Maxwell}. Poisson and later Maxwell favored atomic-sized magnetic charges, always coupled as dipoles, as the fundamental building blocks of magnetic matter \cite{Maxwell}. Amp\`ere, conversely, postulated that magnetism originates from microscopic, atomic-sized closed current loops \cite{Maxwell}. Distinguishing between these two hypotheses is a subtle problem, as highlighted in a seminal 1977 paper by J.~D.~Jackson \cite{Jackson_40}. Ultimately, the key element for distinguishing between the two models is the exact behavior of the magnetic field exerted by a point magnetic dipole at its own site.

This on-site field determines the Fermi contact potential \cite{Fermi} and, consequently, one of the best-measured quantities of relativistic quantum mechanics: the isotropic hyperfine splitting of the $s$-levels in the hydrogen atom \cite{Jackson_40, GriffithsAJP, Jackson}. Given its fundamental importance, the determination of this singular field has long attracted scholarly attention. It is discussed extensively in the classic textbooks of J.~D.~Jackson \cite{Jackson} and D.~J.~Griffiths \cite{Griffiths}, as well as in dedicated pedagogical articles by D.~J.~Griffiths \cite{GriffithsAJP}, C.~P.~Frahm \cite{Frahm}, and Y.~Mu\~niz et al. \cite{Muniz}.

Computing and conceptualizing the magnetic field at the site of a dipole often pushes undergraduate (and even graduate) students to the limits of their mathematical comfort zone. This difficulty arises because the vector functions describing the dipole field exhibit a structural singularity at the origin. Properly managing this singularity precludes the use of standard functions as defined in ordinary calculus, requiring instead the framework of distribution theory or ``generalized functions'' \cite{NIST, Light, Halperin}.

In this paper, we first translate the generalized functions approach into the familiar language of ordinary quantum-mechanical operators. We then demonstrate that a certain matrix central to the problem commutes with all matrices of the three-dimensional rotation group $\mathrm{SO}(3)$. By applying symmetry arguments well-established in quantum mechanics \cite{Wigner, DP}---specifically Schur's Lemma (see, e.g., \cite{DP}, p.~58)---we provide a symmetry-assisted derivation of the on-site magnetic field. Finally, we discuss how the on-site magnetic field becomes anisotropic when the full rotational symmetry is broken down to an axial symmetry. This symmetry-breaking aspect has largely been overlooked in standard textbooks but is highly relevant to contemporary problems, such as two-dimensional hydrogen-like atoms trapped in semiconductor quantum wells \cite{Yang, Xie, Bastard, PRL, Plat}, hyperfine interactions in confined geometries~\cite{Coi, Mer} or two-dimensional ferromagnetism \cite{PIP}.

The bulk of the paper is devoted to presenting the original results, their physical meaning, and their pedagogical relevance. Those results that are mathematically rigorous are stated in the main text in the form of propositions. While these propositions are mathematically rigorous, they are framed with ``pedagogical rigor'' in mind. This ensures that advanced undergraduate and graduate students can look past the abstract formalism and learn how to apply these tools themselves to other problems in physics. The more conceptual proofs are included in the bulk of the paper. More technical proofs are collected in the appendices, one to each. \SM{schur}, giving general instructions about Schur's Lemma and the spectral theorem, comes first, because the division of labour between those two theorems is part of what the paper sets out to teach. We note that while modern computational tools and AI can be utilized to execute standard algebraic tasks, the underlying physical principles and rigorous conceptual thinking must ultimately be mastered through one's own analysis.

\section{Symmetry-assisted computation of the isotropic Fermi-contact magnetic field}

\subsection{Statement of the problem}
The derivation of the magnetic field of a point dipole (see \SM{vecpot}) is based on the expression for the vector potential $\vec{A}(\vec{r})$ at coordinate $\vec{r}$ generated by a magnetic dipole $\vec{\mu}$ located at the origin:
\begin{equation}
\vec{A}(\vec{r}) = \frac{\mu_0}{4\pi} \vec{\nabla} \times \frac{\vec{\mu}}{|\vec{r}|}
\label{Vector}
\end{equation}
In this paper, vectors are defined with respect to a Cartesian coordinate system with coordinates $(x_1, x_2, x_3)$. The $\vec{\nabla}$-operator is given by $(\partial_1, \partial_2, \partial_3)$. In \SMs{vecpot}, we derive this expression explicitly by considering a current circulating around a small closed loop following Amp\`ere's hypothesis, where the magnetic moment $\vec{\mu}$ is defined as the product of the current and the vector area of the loop.

The magnetic field $\vec{B}(\vec{r})$ is subsequently obtained via $\vec{B} = \vec{\nabla} \times \vec{A}$. Applying the standard vector identity $\vec{\nabla} \times (\vec{\nabla} \times \vec{X}) = \vec{\nabla} (\vec{\nabla} \cdot \vec{X}) - \nabla^2 \vec{X}$, the components $B_i$ ($i=1,2,3$) can be written as:
\begin{equation}
B_i(\vec{r}) = \frac{\mu_0}{4\pi}\sum_j \mu_j\partial_i \partial_j \frac{1}{| \vec{r} |}  - \frac{\mu_0}{4\pi}\,\mu_i \nabla^2 \frac{1}{| \vec{r} |}
\label{B}
\end{equation}
The second derivatives $\partial_i \partial_j |\vec{r}|^{-1}$ and the Laplacian $\nabla^2 |\vec{r}|^{-1}$ are highly singular at the origin where the dipole resides (see \SM{singular} for a brief discussion). To properly evaluate this singularity, these functions must be treated as the kernel of a test integral that defines a functional (called a ``distribution'' or ``generalized function'') \cite{NIST}:
\begin{equation}
\int \mathrm{d}V f(\vec{r}) B_i(\vec{r}) = \frac{\mu_0}{4\pi} \sum_j \mu_j \left[ \int \mathrm{d}V f(\vec{r}) \partial_i \partial_j \frac{1}{|\vec{r}|} \right] - \frac{\mu_0}{4\pi} \mu_i \int \mathrm{d}V f(\vec{r}) \nabla^2 \frac{1}{|\vec{r}|}
\label{Testintegral}
\end{equation}
In this formulation, $f(\vec{r})$ represents a ``test function'' \cite{NIST}, i.e., a smooth, well-behaved function that ensures the convergence of the integral by decaying rapidly at infinity. The $s$-wave function of the hydrogen atom serves as a physical example of such a test function, and the radial matrix element in the hyperfine splitting calculation \cite{GriffithsAJP} is precisely a test integral of this type. The classic isotropic hyperfine splitting is typically extracted by evaluating equation \eqref{Testintegral} using a spherically symmetric test function.

In equation \eqref{Testintegral}, we single out the test-integral matrices $F$ and $G$, defined by their matrix elements:
\begin{equation}
F_{ij} \equiv \int \mathrm{d}V f(|\vec{r}|) \partial_i \partial_j \frac{1}{|\vec{r}|}, \quad G_{ij} \equiv \delta_{ij} \int \mathrm{d}V f(|\vec{r}|) \nabla^2 \frac{1}{|\vec{r}|}
\label{FG_matrices}
\end{equation}
In the language of quantum mechanics, these test integrals---namely, the matrix elements of $F$ and $G$---can be viewed as scalar products involving the differential operators $\partial_i \partial_j$, an ``initial'' state $|\vec{r}|^{-1}$, and a ``final'' state $f(|\vec{r}|)$.

\subsection{Symmetry-based results}
Let us now analyze this matrix framework using the representation theory of groups \cite{Wigner, DP}.

\begin{proposition}
\label{prop:commute}
The matrix $F$ (and trivially the diagonal matrix $G$) commutes with any $3\times 3$ rotation matrix $R \in \mathrm{SO}(3)$ that transforms the coordinate system $(x_1, x_2, x_3)$, i.e.,
\begin{equation}
FR = RF.
\end{equation}
\end{proposition}

\begin{proof}
See \SM{commute}.
\end{proof}

\noindent To understand why this algebraic property is central to resolving the on-site field problem, we state the following propositions.

\begin{proposition}
\label{prop:isotropy}
By virtue of $FR = RF$ for all $R \in \mathrm{SO}(3)$, rotating the magnetic moment $\vec{\mu}$ rotates the on-site magnetic field vector by the exact same angle and direction, directly demonstrating the isotropy of the Fermi contact magnetic field.
\end{proposition}

\begin{proof}
Consider the on-site field expressed as $\frac{\mu_0}{4\pi}(F-G)\vec{\mu}$, with $F$ and $G$ as defined in equation \eqref{FG_matrices}. If we rotate the dipole moment such that $\vec{\mu} \rightarrow R\vec{\mu}$, the resulting on-site field becomes $\frac{\mu_0}{4\pi}(F-G)(R\vec{\mu})$. Because $F$ and $G$ commute with $R$, we can rearrange the terms as:
\begin{equation}
\frac{\mu_0}{4\pi}(F-G)(R\vec{\mu}) = R \left[ \frac{\mu_0}{4\pi}(F-G)\vec{\mu} \right].
\end{equation}
The right-hand side is precisely the original on-site field rotated by the matrix $R$.
\end{proof}

\noindent The commutativity of $F$ with $R$ carries an even deeper consequence.

\begin{proposition}
\label{prop:scalar}
The matrix $F$ is a scalar multiple of the identity matrix: it is diagonal, with identical real entries along its main diagonal, just as $G$ is by definition.
\end{proposition}

\begin{comment}
This proposition vastly simplifies the problem of computing the matrix $F$: instead of computing nine individual components, we only need to evaluate one single non-vanishing matrix element, which we take to be $F_{33}$.
\end{comment}

\begin{proof}
The proof proceeds in two steps, which are discussed at greater length in Appendix~\ref{app:schur}.

\textit{1. Schur's Lemma.} The $3\times 3$ rotation matrices in Euclidean space form the standard irreducible representation of the rotation group $\mathrm{SO}(3)$. Schur's Lemma dictates that any matrix commuting with all matrices of an irreducible representation must be a scalar multiple of the identity matrix $I$, so that $FR = RF$ for every $R \in \mathrm{SO}(3)$ gives
\begin{equation}
F = \lambda \, I.
\label{Schur_scalar}
\end{equation}
The lemma is a statement about representations on complex vector spaces, and it is on $\mathbb{C}^3$ that it is applied here; the scalar $\lambda$ it delivers is in general a complex number.

\textit{2. The spectral theorem.} By its definition \eqref{FG_matrices} the matrix $F$ is real and symmetric, since $\partial_i\partial_j = \partial_j\partial_i$. A real symmetric matrix is Hermitian, and the spectral theorem of linear algebra guarantees that the eigenvalues of a Hermitian matrix are real. By \eqref{Schur_scalar} the scalar $\lambda$ is the eigenvalue of $F$, and it is therefore a real number.

The matrix $F$ is thus diagonal, with three identical real entries.
\end{proof}

\noindent Proposition~\ref{prop:scalar} reduces the problem to a single number, but that number must be computed ``by brute force'': no symmetry argument is available for it. The following proposition deals with the computation of $F_{33}$.

\begin{proposition}
\label{prop:F33}
The diagonal matrix element $F_{33}$ satisfies
\begin{equation}
F_{33} = -\frac{4\pi}{3} f(0).
\label{Prop4_Eq}
\end{equation}
\end{proposition}

\begin{comment}
When this result is substituted back into equation \eqref{Testintegral}, the total test integral, including the matrix elements of $G$, yields the standard isotropic Fermi contact magnetic field. The arithmetic is worth displaying, because it is the cancellation between the two terms of equation \eqref{B} that fixes the celebrated factor $2/3$:
\begin{equation}
\frac{\mu_0}{4\pi}\Big[F-G\Big]=
\begin{pmatrix}
\frac{2\mu_0}{3}f(0)&0&0\\
0&\frac{2\mu_0}{3}f(0)&0\\
0&0&\frac{2\mu_0}{3}f(0)
\end{pmatrix}
\end{equation}
or, simply,
\begin{equation}
\vec{B}_{\text{on-site}}(\vec{r}) = \frac{2\mu_0}{3} \vec{\mu} \, \delta(\vec{r}).
\label{Fermi_contact}
\end{equation}
\end{comment}

\begin{proof}
For the proof, two elements are essential. The first pertains to the definition of $F_{33}$. Following the framework of distribution theory \cite{Halperin, NIST}, the distributional derivative of $F_{33}$,
\begin{equation}
F_{33} = \int \mathrm{d}V f(|\vec{r}|) \partial_3^2 \left(\frac{1}{|\vec{r}|}\right),
\end{equation}
is defined by transferring the derivatives onto the test function:
\begin{equation}
F_{33} = \int \mathrm{d}V \frac{1}{|\vec{r}|} \partial_3^2 f(|\vec{r}|).
\end{equation}
The second element pertains to the evaluation of the integral over the volume. Because the origin is a singular point, the volume integral must be handled by isolating the origin within a small sphere of radius $\epsilon$ and taking the limit $\epsilon \to 0$:
\begin{equation}
\int \mathrm{d}V \frac{1}{|\vec{r}|} \partial_3^2 f(|\vec{r}|) = \lim_{\epsilon\rightarrow 0} \int_{|\vec{r}|>\epsilon} \mathrm{d}V \frac{1}{|\vec{r}|} \partial_3^2 f(|\vec{r}|).
\label{F33}
\end{equation}
The details of the proof are given in \SM{F33}.
\end{proof}

\section{Symmetry-assisted computation of the anisotropic Fermi-contact magnetic field}
\label{sec:anisotropic}

The standard expressions for the on-site magnetic field, and with them the isotropy established in Proposition~\ref{prop:isotropy}, are derived assuming spherically symmetric test functions. While these functions are perfectly suited for three-dimensional problems like atomic hyperfine structures, significant physical scenarios exist where this spatial symmetry is broken. In the following, we consider physical problems where the test functions are expressed as a product $g(\rho)h(x_3)$, where $\rho = \sqrt{x_1^2 + x_2^2}$ is the in-plane radial coordinate. Here, $g(\rho)$ describes the in-plane distribution and $h(x_3)$ represents the test function defined along the coordinate $x_3$. A prime example occurs when one spatial degree of freedom (e.g., along the $x_3$-axis) is suppressed or strongly confined, giving rise to an effectively two-dimensional system, such as hydrogen-like atoms trapped inside narrow semiconductor quantum wells \cite{Xie, Bastard, PRL}. In these two-dimensional hydrogenic systems \cite{Yang}, the electron experiences a three-dimensional Coulomb potential, but its motion is constrained to a plane because the third spatial degree of freedom is energetically frozen. Any physically realistic test function for this geometry must reflect this directional anisotropy.

The matrix $F$ constructed over this class of product test functions no longer possesses full rotational symmetry; instead, it only commutes with a subgroup of the full rotation group, namely the axial group consisting of all rotations $R(\varphi)$ about the $x_3$-axis. Their matrix representation in the $(x_1, x_2, x_3)$ coordinate system reads:
\begin{equation}
R(\varphi) = \begin{pmatrix}
\cos \varphi & -\sin\varphi & 0 \\
\sin\varphi &  \cos\varphi & 0 \\
0 & 0 & 1
\end{pmatrix}.
\end{equation}
The two-dimensional matrix representation of the rotation in the $(x_1, x_2)$-plane (the top-left block in the matrix) is, over the complex numbers on which Schur's Lemma operates, reducible, so that the lemma does not directly apply. Nevertheless, the commutativity of the matrix $F$ with $R(\varphi)$ yields the following results: the two ingredients used for Proposition~\ref{prop:scalar}---Schur's Lemma and the spectral theorem, both discussed in Appendix~\ref{app:schur}---settle this case as well, once the reducibility is taken into account.

\begin{proposition}
\label{prop:axial}
The matrix $F$ is diagonal, with diagonal matrix elements satisfying $F_{11} = F_{22} \equiv F_{\parallel}$, which need no longer coincide with $F_{33} \equiv F_{\perp}$.
\end{proposition}

\begin{comment}
Symmetry alone permits $F_\parallel \neq F_\perp$; the explicit evaluation below shows that the two do differ, which removes the structural isotropy of the Fermi contact potential.
\end{comment}

\begin{proof}
We first show that the unit vector $\vec{e}_3$ pointing along the $x_3$-coordinate is an eigenvector of $F$. Indeed, because $R(\varphi) \vec{e}_3 = \vec{e}_3$, we have:
\begin{equation}
F R(\varphi) \vec{e}_3 = R(\varphi) F \vec{e}_3 = F \vec{e}_3.
\end{equation}
Accordingly, $F\vec{e}_3$ is a vector that remains invariant under axial rotation, which implies:
\begin{equation}
F \vec{e}_3 = F_{33} \vec{e}_3.
\end{equation}
This result ensures that the symmetric matrix $F$ possesses the same block structure as the rotation matrix $R(\varphi)$:
\begin{equation}
F = \begin{pmatrix}
a & b & 0 \\
b & c & 0 \\
0 & 0 & F_{33}
\end{pmatrix}.
\end{equation}
\medskip
The commutativity between the in-plane rotations and the upper-left $2 \times 2$ block of $F$---denoted as $F^\parallel$---imposes a strict condition on the structure of this block. Over the complex numbers, the representation
$R(\varphi)$ of the rotations about the $x_3$-axis is not irreducible: the two vectors
\begin{equation}
\vec{v}_\pm = \begin{pmatrix} 1 \\ \mp \mathrm{i} \end{pmatrix}
\qquad \text{satisfy} \qquad
R(\varphi)\, \vec{v}_\pm = \mathrm{e}^{\pm \mathrm{i}\varphi} \, \vec{v}_\pm.
\label{circular}
\end{equation}
The in-plane representation is therefore reducible and splits into the two one-dimensional representations $\mathrm{e}^{\pm\mathrm{i}\varphi}$. Schur's Lemma does not apply to it. Physically, $\vec{v}_\pm$ are the circular polarizations of the plane, and it is entirely reasonable that an axially symmetric problem should single them out.

\medskip
Nevertheless, group theory provides us with enough knowledge to find the structure of $F^{\parallel}$. In fact, \textit{i.} because of the commutativity of $R(\varphi)$ and $F^\parallel$, the complex vectors $\vec{v}_\pm$ of equation \eqref{circular} are also eigenvectors of $F^\parallel$! Furthermore, \textit{ii.} the spectral theorem requires the eigenvalues $\alpha^{\pm}$ of $F^\parallel$ to be real numbers. \textit{i.} and \textit{ii.} finally determine the structure of $F^\parallel$. In fact, the equation
\begin{equation}
\begin{pmatrix}
a & b\\
b & c
\end{pmatrix}
\begin{pmatrix} 1 \\ - \mathrm{i} \end{pmatrix}= \alpha^+ \begin{pmatrix} 1 \\ - \mathrm{i} \end{pmatrix}
\end{equation}
has the solution $b=0$ and $a=c$, i.e.\
\begin{equation}
F^\parallel = F_\parallel \, I.
\end{equation}

\end{proof}

\begin{comment}
The matrix elements follow from transferring the derivatives onto the test function. Since $\nabla^2_\rho = \partial_1^2 + \partial_2^2$, the in-plane Laplacian delivers the \textit{sum} of the two equal in-plane elements,
\begin{equation}
F_{11} + F_{22} = 2 F_\parallel = \int \mathrm{d}V \frac{1}{\sqrt{\rho^2 + x_3^2}} \, h(x_3) \, \nabla^2_\rho g(\rho),
\label{F_parallel_sum}
\end{equation}
so that, carrying out the trivial azimuthal integration of $\mathrm{d}V = \rho \,\mathrm{d}\rho \,\mathrm{d}\varphi \,\mathrm{d}x_3$,
\begin{equation}
F_\parallel = \pi \int \mathrm{d}x_3 \, h(x_3) \int \rho \,\mathrm{d}\rho \, \frac{1}{\sqrt{\rho^2 + x_3^2}} \, \nabla^2_\rho g(\rho).
\label{F_parallel_def}
\end{equation}
For $F_\perp$, we use the identity
\begin{equation}
\partial_3^2 \equiv \nabla^2 - \nabla^2_\rho
\end{equation}
to write, with the help of equation \eqref{F_parallel_sum},
\begin{equation}
F_\perp = \int \mathrm{d}V \frac{1}{\sqrt{\rho^2 + x_3^2}} \nabla^2 [g(\rho) h(x_3)] - 2 F_\parallel.
\label{F_perp_def}
\end{equation}
Expressing $F_\perp$ in this manner allows one to identify the test integral of the three-variable Laplacian of $|\vec{r}|^{-1}$ and $F_\parallel$ as the sought-for quantities.
\end{comment}

\begin{proposition}
\label{prop:laplacian2d}
For the product test functions $h(x_3) g(\rho)$ considered here, the test integral of the three-dimensional Laplacian of $|\vec{r}|^{-1}$ is
\begin{equation}
\int \nabla^2 \left( \frac{1}{\sqrt{\rho^2 + x_3^2}} \right) h(x_3) g(\rho) \,\mathrm{d}V = -4\pi h(0) g(0).
\label{Prop6_Eq}
\end{equation}
\end{proposition}

\begin{comment}
Equation \eqref{Prop6_Eq} is the familiar identity $\nabla^2 |\vec{r}|^{-1} = -4\pi \delta(\vec{r})$. What the proof in \SM{prop6} shows is that the identity survives when the test functions are chosen with cylindrical instead of spherical symmetry and the origin is excluded, accordingly, using a small \textit{cylinder} rather than a small sphere. The test integral is found to be independent of the aspect ratio with which the origin is approached.
\end{comment}

\begin{proof}
See \SM{prop6}.
\end{proof}

\begin{comment}
These results might be relevant for determining the Fermi contact potential in two-dimensional hydrogen like systems. Note that existing literature on two-dimensional hydrogen atoms (e.g., \cite{Yang}) typically enforces two-dimensionality by eliminating the $x_3$ variable from the underlying differential equations entirely. In \SM{slab} we discuss a slab model. This model assumes a class of smooth, symmetric test functions satisfying $h(x_3) = h(-x_3)$, which are positive and possess a characteristic spatial scale $d$. The in-plane test functions $g(\rho)$ are smooth, positive, and possess a spatial scale $\Lambda$. The slab model keeps all three spatial coordinates active, allowing us to seamlessly apply standard theorems of vector calculus. Two-dimensionality is then enforced by imposing the scale separation $d/\Lambda \rightarrow 0$. The results of \SMs{slab} support the \textbf{conjecture} that, in this limit, $F_\parallel=0$, i.e.
\begin{equation}
\frac{\mu_0}{4\pi}\Big[F-G\Big]=
\begin{pmatrix}
\mu_0 h(0)g(0)&0&0\\
0&\mu_0 h(0)g(0)&0\\
0&0&0
\end{pmatrix}
\end{equation}
This indicates that the on-site magnetic field is strongly anisotropic. For instance, in a configuration where the magnetic dipole is aligned strictly along the $x_3$-axis, $\vec{B}_{\text{on-site}}(\vec{r})=0$. Conversely, if the magnetic moment lies entirely within the confinement plane, $\vec{B}_{\text{on-site}}(\vec{r})=\mu_0 \vec{\mu} \, \delta(\vec{r})$.
\end{comment}

\section{A theorem of classical magnetostatics for test integrals and its application to two-dimensional ferromagnets}

Another highly practical scenario that can be elegantly resolved using the symmetry-breaking arguments developed above concerns two-dimensional ferromagnets. These systems are realization-dependent and typically obtained either as ultrathin transition-metal films deposited on non-magnetic substrates or as isolated monolayers of exfoliated bulk crystals \cite{PIP}. Such systems carry a distribution of magnetic moments strictly confined to a planar geometry. To analyze this class of low-dimensional systems, we establish a theorem that bridges the abstract distribution-theoretic test integral directly with classical magnetostatics.

\begin{hypothesis}
Following Jackson \cite{Jackson} (p.~197), the macroscopic magnetic field $\vec{B}(\vec{r})$ generated by a continuous, bounded distribution of magnetic moments described by a magnetization vector field $\vec{M}(\vec{r})$ is given by:
\begin{equation}
\vec{B}(\vec{r}) = \frac{\mu_0}{4\pi} \vec{\nabla} \times \int \mathrm{d}V' \frac{\vec{\nabla}' \times \vec{M}(\vec{r}')}{|\vec{r} - \vec{r}'|},
\end{equation}
where the magnetization field $\vec{M}(\vec{r})$ represents the local magnetic moment per unit volume. In a discrete lattice with lattice constant $a$, where each localized magnetic dipole $\vec{\mu}$ occupies an atomic volume $a^3$, the magnetization can be approximated as $\vec{M}(\vec{r}) = \vec{\mu}(\vec{r})/a^3$. Here, the primed operator $\vec{\nabla}'$ acts exclusively on the source coordinates $\vec{r}'$.
\end{hypothesis}

Let us consider a specialized magnetization profile where a constant magnetic moment vector $\vec{\mu}$ is spatially modulated by a continuous, localized distribution function $f(\vec{r})$, such that $\vec{M}(\vec{r}) = \vec{\mu} f(\vec{r})$.

\begin{proposition}
\label{prop:magnetostatic}
Under these conditions, the macroscopic magnetic field evaluated at the origin $\vec{r} = \vec{0}$ is given exactly by:
\begin{equation}
\vec{B}(\vec{0}) = \frac{\mu_0}{4\pi} \int \mathrm{d}V f(\vec{r}) \vec{\nabla} \left( \sum_j \mu_j \partial_j \frac{1}{|\vec{r}|} \right) + \mu_0 \vec{\mu} f(\vec{0}).
\label{Magneto}
\end{equation}
\end{proposition}

\begin{proof}
The derivation utilizes standard vector identity expansions and integration-by-parts techniques. The complete step-by-step verification is detailed in \SM{magneto}.
\end{proof}

\noindent The mathematical expression on the right-hand side of equation \eqref{Magneto} is formally identical to the distribution-theoretic test integral of the magnetic field of a single localized point dipole $\vec{\mu}$ defined in equation \eqref{B}. This equivalence translates the abstract concept of a test integral of $\vec{B}$ into a concrete, observable macroscopic field value within classical magnetostatics.

As a direct application of this theorem, we utilize the previously calculated values for the test integrals $F_{\perp}$ and $F_{\parallel}$ of equations \eqref{F_parallel_def} and \eqref{F_perp_def}, together with $G_{ii}$ from Proposition~\ref{prop:laplacian2d}, and the diagonal structure guaranteed by Proposition~\ref{prop:axial}, to analyze a uniformly magnetized ultrathin slab with constant magnetization $\vec{M}$. Because the continuous in-plane translational symmetry allows us to arbitrarily place the origin of the coordinate system anywhere within the plane of the slab, the magnetic field computed at $\vec{r} = \vec{0}$ is uniform throughout the entire volume of the layer.

It is instructive to examine the perpendicular-to-plane magnetization configuration in more detail. The field at any site consists of three components:
\begin{itemize}
\item The first originates from the matrix $F$ and amounts to $-\mu_0 h(0) g(0)$. It represents the part of the magnetic field that a dipole exerts upon itself; we call it a ``local'' component.
\item The second, also originating from the $F$-matrix, is $-\frac{\mu_0}{2\pi} F_\parallel$. It represents the magnetic field that all neighbouring magnetic moments, the distant ones included, exert on the magnetic moment at the site under consideration. It is, accordingly, a ``non-local'' component.
\item The third is again local and originates from the $G$-matrix: $+\mu_0 h(0) g(0)$.
\end{itemize}
These results allow for a striking physical prediction. Provided the slab is infinitely extended, i.e.\ $\Lambda \to \infty$, the quantity $F_\parallel$ vanishes and the two local components \textit{cancel out} exactly, so that the $\vec{B}$ field vanishes. If the slab is finite, the local components still cancel, but the non-local one---however weak---survives.

Experimental observations~\cite{Yves} appear to confirm this prediction. Those experiments attempted to inject a uniform perpendicular magnetic field into a circular, micron-sized dot of nanometre thickness. The field did enter the dot, but not uniformly: it penetrated first at the rim, and only then spread over the whole of the disc. The non-local component also plays a major role whenever spatially non-uniform spin configurations in infinitely extended, perpendicularly magnetized systems are considered~\cite{PIP}.

\section{Conclusion}
We close this paper with a few concluding remarks. One pertains to the vector field components given by $\frac{1}{4\pi} \sum_j \partial_i \partial_j |\vec{r}|^{-1} \mu_j$ (the first term in equation \eqref{B}). This component represents the $\vec{H}(\vec{r})$ field of the magnetic dipole, which is an auxiliary vector field commonly used in magnetostatics. The $\vec{H}(\vec{r})$ field of the magnetic dipole possesses a particularly significant property in the context of this work. Specifically, the quantity $\mu_0 \vec{H}$ represents the magnetic field produced by a pair of fictitious magnetic charges within Poisson's model. As equation \eqref{B} demonstrates, the field arising from $\vec{H}$ is identical to the true Amp\`erian $\vec{B}$-field everywhere except at the exact site of the dipole. This subtle, highly localized difference explains why distinguishing between the Poisson--Maxwell and Amp\`ere hypotheses experimentally required such high-precision atomic physics.

In summary, we have reframed the long-standing problem of the singular on-site magnetic field of a localized point dipole by mapping the distribution-theoretic test integrals onto the matrix representation of differential operators. By exploiting the spatial symmetry groups of the problem, we demonstrated how Schur's Lemma can be utilized as a powerful computational aid to evaluate these highly singular fields without resorting to cumbersome brute-force integration.

Furthermore, we extended this formalism to a symmetry-broken environment, demonstrating how spatial confinement within a two-dimensional slab renders the traditional Fermi contact potential highly anisotropic. Finally, we proved a magnetostatic theorem that establishes a direct equivalence between these abstract test integrals and the macroscopic fields of continuous magnetization distributions, applying it successfully to ultrathin magnetic films.

The primary pedagogical and scientific conclusion of this work is that group-theoretic symmetry arguments---which students typically associate exclusively with quantum mechanics or solid-state physics---can be seamlessly applied to classical magnetostatics. This approach provides advanced undergraduate and graduate students with a unified, elegant conceptual tool to master complex mathematical singularities in physical systems.

\section*{Organisation of the appendices}
The two theorems on which the symmetry argument rests are set out in
appendix~\ref{app:schur}. The derivations supporting the more technical
propositions follow, one to an appendix:
appendix~\ref{app:vecpot}, the vector potential of a small current loop;
appendix~\ref{app:singular}, the singular behaviour of $\partial_3^2 |\vec{r}|^{-1}$;
appendix~\ref{app:commute}, the proof of Proposition~\ref{prop:commute};
appendix~\ref{app:F33}, the proof of Proposition~\ref{prop:F33};
appendix~\ref{app:prop6}, the proof of Proposition~\ref{prop:laplacian2d};
appendix~\ref{app:slab}, the slab model supporting the conjecture $F_\parallel \to 0$;
and appendix~\ref{app:magneto}, the proof of Proposition~\ref{prop:magnetostatic}.

\appendix
\numberwithin{equation}{section}

\section{Schur's Lemma and the spectral theorem}
\label{app:schur}

The proof of Proposition~\ref{prop:scalar} rests on two theorems drawn from different parts of mathematics: Schur's Lemma, which belongs to the representation theory of groups, and the spectral theorem, which belongs to linear algebra. Each contributes something the other cannot, and the division of labour between them is worth setting out.

\medskip
\noindent \textit{1. Irreducible representations.} A set of matrices representing a group acts on a vector space $V$. The representation is called \textit{irreducible} when no subspace of $V$ other than $\{\vec{0}\}$ and $V$ itself is carried into itself by every matrix of the set. The $3\times3$ rotation matrices acting on Euclidean space furnish the standard example. No line survives: a rotation through $\pi/2$ about an axis perpendicular to a given line carries that line elsewhere. No plane survives either, for rotations preserve the scalar product, so an invariant plane would be accompanied by an invariant perpendicular line, which has just been excluded.

\medskip
\noindent \textit{2. Schur's Lemma.} The lemma states that a matrix $A$ commuting with every matrix of an irreducible representation is a multiple of the identity \cite{DP}. Its proof is short. Let $\lambda$ be an eigenvalue of $A$ and consider
\begin{equation}
V_\lambda = \{ \vec{v} \in V : A\vec{v} = \lambda \vec{v} \},
\end{equation}
which is not the zero space. For any matrix $R$ of the representation, $A(R\vec{v}) = R(A\vec{v}) = \lambda R\vec{v}$, so $R$ carries $V_\lambda$ into itself. As the representation is irreducible, $V_\lambda$ must be the whole of $V$, which is to say $A = \lambda I$.

Notice that the argument opens by asserting that $A$ \textit{has} an eigenvalue. This is guaranteed only over the complex numbers, where the characteristic polynomial is certain to have a root. Schur's Lemma is for this reason a statement about representations on complex vector spaces, and to invoke it for the real matrix $F$ one regards $F$ and the rotation matrices as acting on $\mathbb{C}^3$ rather than on $\mathbb{R}^3$. The matrices are unchanged and the representation remains irreducible, but the scalar the lemma returns is a priori a complex number: representation theory alone cannot say otherwise.

\medskip
\noindent \textit{3. The spectral theorem.} What restricts $\lambda$ is a property of $F$ that has nothing to do with rotations. By its definition \eqref{FG_matrices} the matrix is real and symmetric, $F_{ij} = F_{ji}$, since the order of the two partial derivatives is immaterial. A real symmetric matrix is Hermitian, and the spectral theorem asserts that a Hermitian matrix has real eigenvalues and an orthonormal basis of eigenvectors. It is the same theorem that guarantees real expectation values for the observables of quantum mechanics. Since $\lambda$ is the eigenvalue of $F = \lambda I$, it is real, and Proposition~\ref{prop:scalar} is established.

\section{The vector potential of a small current loop}
\label{app:vecpot}
In Ref.~\cite{Jackson} (pp.~184--186), the vector potential for a single magnetic dipole is obtained from a volume integral containing a steady current density. Here, we follow an alternative approach, which is less abstract and directly based on Amp\`ere's hypothesis. We consider a thin wire forming a small, closed loop $C$ that bounds a surface $S$. For such a thin wire, the volume integral over the current density can be rewritten as a line integral along the loop $C$ and subsequently transformed, via Stokes' theorem, into an integral over the surface $S$:
\begin{equation}
\vec{A}(\vec{r}) = \frac{\mu_0 I}{4\pi} \oint_C \frac{1}{|\vec{r}-\vec{r}'|} \,\mathrm{d}\vec{s}' = \frac{\mu_0 I}{4\pi} \int_{S} \vec{n} \times \vec{\nabla}' \left( \frac{1}{|\vec{r}-\vec{r}'|} \right) \mathrm{d}S'.
\label{A_App}
\end{equation}
In equation \eqref{A_App}, the coordinates along the loop are denoted by $\vec{r}'$, $\mathrm{d}\vec{s}'$ is the vector line element along the path, $\vec{n}$ is the unit normal vector to the surface $S$, and $\vec{\nabla}'$ denotes the gradient operator with respect to the source coordinates $\vec{r}'$.

Evaluating the gradient with respect to the primed coordinates via the chain rule yields:
\begin{equation}
\vec{\nabla}' \left( \frac{1}{|\vec{r}-\vec{r}'|} \right) = \frac{\vec{r}-\vec{r}'}{|\vec{r}-\vec{r}'|^3}.
\label{Eq:Nabla_App}
\end{equation}
For an ideal point dipole located at the origin, the spatial dimensions of the loop shrink to zero ($|\vec{r}'| \to 0$), while the observation point remains far away ($|\vec{r}| \gg |\vec{r}'|$). We can approximate equation \eqref{Eq:Nabla_App} to first order by setting $\vec{r}' \approx \vec{0}$:
\begin{equation}
\vec{\nabla}' \left( \frac{1}{|\vec{r}-\vec{r}'|} \right) \approx \frac{\vec{r}}{|\vec{r}|^3} + \mathcal{O}\left(\frac{|\vec{r}'|}{|\vec{r}|^3}\right).
\end{equation}
Substituting this expansion back into the surface integral, the term $\vec{r}/|\vec{r}|^3$ depends solely on the unprimed observation coordinate and can be factored outside of the integration:
\begin{equation}
\vec{A}(\vec{r}) \approx \frac{\mu_0 I}{4\pi} \left( \int_{S} \vec{n} \,\mathrm{d}S' \right) \times \frac{\vec{r}}{|\vec{r}|^3}.
\end{equation}
We define the magnetic dipole moment $\vec{\mu}$ as the product of the current and the vector area of the loop:
\begin{equation}
\vec{\mu} = I \int_{S} \vec{n} \,\mathrm{d}S' = I S \vec{n}.
\end{equation}
Using this definition, the classic expression for the vector potential of a point magnetic dipole takes the familiar form:
\begin{equation}
\vec{A}(\vec{r}) = \frac{\mu_0}{4\pi} \frac{\vec{\mu} \times \vec{r}}{|\vec{r}|^3} = \frac{\mu_0}{4\pi} \vec{\nabla} \times \frac{\vec{\mu}}{|\vec{r}|}.
\end{equation}
The expression on the right-hand side is established directly using a standard identity of vector analysis, and is equation \maeq{Vector}.

\section{Singular behavior of $\partial_3^2 |\vec{r}|^{-1}$}
\label{app:singular}
We discuss one of the nine partial derivatives in equation \maeq{B} to illustrate the singular behavior at the origin. For any $\vec{r} \neq \vec{0}$, explicit differentiation yields:
\begin{equation}
\partial_3^2 \left(\frac{1}{|\vec{r}|}\right) = \frac{3x_3^2 - |\vec{r}|^2}{|\vec{r}|^5} = \frac{2x_3^2 - x_1^2 - x_2^2}{(x_1^2 + x_2^2 + x_3^2)^{5/2}}.
\end{equation}
Considering the algebraic expression on the right-hand side as a standard ``point'' function, one would be tempted to evaluate its value at $\vec{r} = \vec{0}$ by computing its limit when the variables approach zero. For example, let us approach the origin strictly along the $x_3$-axis. We obtain:
\begin{equation}
\lim_{\substack{x_3 \to 0 \\ x_1=x_2=0}} \frac{2x_3^2 - x_1^2 - x_2^2}{(x_1^2 + x_2^2 + x_3^2)^{5/2}} = \lim_{x_3 \to 0} \frac{2x_3^2}{|x_3|^5} = \lim_{x_3 \to 0} \frac{2}{|x_3|^3} = \infty.
\end{equation}
However, approaching the origin along the diagonal line $x_1 = x_2 = x_3$ yields a completely different result:
\begin{equation}
\lim_{\substack{x_3 \to 0 \\ x_1=x_2=x_3}} \frac{2x_3^2 - x_1^2 - x_2^2}{(x_1^2 + x_2^2 + x_3^2)^{5/2}} = \lim_{x_3 \to 0} \frac{0}{(3x_3^2)^{5/2}} = 0.
\end{equation}
Because the limit depends heavily on the direction of approach, the second derivative $\partial_3^2 |\vec{r}|^{-1}$ is ill-defined at $\vec{r} = \vec{0}$ when treated as an ordinary function in standard calculus.

\section{Proof of commutativity of $F$ with a rotation matrix $R$ (\maprop{prop:commute})}
\label{app:commute}

\noindent \textit{1. Coordinate Transformation:} We consider a rotation of the coordinates defined by $\vec{x} = R\vec{y}$. Since $R$ is an orthogonal rotation matrix, $R^T = R^{-1}$. In index notation, this reads:
\begin{equation}
x_i = \sum_{k} R_{ik} y_k \quad \text{and} \quad \frac{\partial x_i}{\partial y_k} = R_{ik}.
\end{equation}
Due to the orthogonality of $R$, the inverse transformation similarly yields:
\begin{equation}
y_k = \sum_{i} R_{ik} x_i \quad \text{and} \quad \frac{\partial y_k}{\partial x_i} = R_{ik}.
\end{equation}

\noindent \textit{2. Transformation of Partial Derivatives:} We transform the partial derivatives from the $\vec{x}$-system to the $\vec{y}$-system via the chain rule. For the first derivative, we have:
\begin{equation}
\frac{\partial}{\partial x_i} = \sum_{k} \frac{\partial y_k}{\partial x_i} \frac{\partial}{\partial y_k} = \sum_{k} R_{ik} \frac{\partial}{\partial y_k}.
\end{equation}
For the second partial derivatives (the components of the Hessian matrix), applying this relation twice yields:
\begin{equation}
\frac{\partial^2}{\partial x_i \partial x_j} = \left(\sum_{k} R_{ik} \frac{\partial}{\partial y_k}\right) \left(\sum_{m} R_{jm} \frac{\partial}{\partial y_m}\right) = \sum_{k, m} R_{ik} R_{jm} \frac{\partial^2}{\partial y_k \partial y_m}.
\end{equation}

\noindent \textit{3. Integral Substitution:} Because the functions $r^{-1}$ and $f(r)$ depend solely on the rotationally invariant radius $r = |\vec{x}| = |\vec{y}|$, their functional forms remain identical in the transformed coordinate system. Furthermore, since the determinant of the Jacobian matrix for an orthogonal transformation is unity, the volume element is invariant, i.e., $\mathrm{d}^3x = \mathrm{d}^3y$. Substituting these transformed derivatives into the definition of $F_{ij}$ gives:
\begin{equation}
F_{ij} = \int_{\mathbb{R}^3} \left( \sum_{k, m} R_{ik} R_{jm} \frac{\partial^2}{\partial y_k \partial y_m}\left(\frac{1}{r}\right) \right) f(r) \,\mathrm{d}^3y.
\end{equation}

\noindent \textit{4. Factoring the Rotation Matrices:} Since the components of the matrix $R$ are constant coefficients, we can pull the summations and the matrix elements outside of the integral:
\begin{equation}
F_{ij} = \sum_{k, m} R_{ik} R_{jm} \left( \int_{\mathbb{R}^3} f(r) \frac{\partial^2}{\partial y_k \partial y_m}\left(\frac{1}{r}\right) \,\mathrm{d}^3y \right).
\end{equation}
The remaining definite integral over the entire space $\mathbb{R}^3$ depends solely on the dummy integration variables $\vec{y}$. Renaming these integration variables back to $\vec{r}$ reveals that the integral corresponds exactly to the definition of the matrix component $F_{km}$:
\begin{equation}
F_{ij} = \sum_{k, m} R_{ik} F_{km} R_{jm}.
\end{equation}

\noindent \textit{5. Conversion to Matrix Notation:} We rewrite this index equation in compact matrix form. Utilizing the relation $R_{jm} = (R^T)_{mj}$, we obtain:
\begin{equation}
F_{ij} = \sum_{k, m} R_{ik} F_{km} (R^T)_{mj} \implies F = R F R^T.
\end{equation}
Because $R^T = R^{-1}$ holds for any rotation matrix, multiplying from the right by $R$ immediately delivers the desired commutation relation $FR = RF$. \hfill $\blacksquare$

\section{Proof of \maprop{prop:F33} (equation \maeq{Prop4_Eq})}
\label{app:F33}
Since $f(|\vec{r}|)$ depends only on the radial coordinate, we evaluate the test integral \maeq{F33} in spherical coordinates using standard conventions:
\begin{equation}
x_1 = r \sin\theta \cos\phi, \quad x_2 = r \sin\theta \sin\phi, \quad x_3 = r \cos\theta.
\end{equation}
Transforming the operator $\partial_3^2$ and retaining only the terms containing radial derivatives (as the operator acts on a spherically symmetric function), we obtain:
\begin{equation}
\frac{\partial^2}{\partial x_3^2} = \cos^2\theta \frac{\partial^2}{\partial r^2} + \frac{\sin^2\theta}{r} \frac{\partial}{\partial r}.
\end{equation}
The test integral thus splits into the sum of two separate components:
\begin{equation}
F_{33} = \lim_{\epsilon\rightarrow 0} \int_\epsilon^{\infty} \mathrm{d}r \, r \frac{\partial^2 f(r)}{\partial r^2} \int \mathrm{d}\Omega \cos^2\theta + \lim_{\epsilon\rightarrow 0} \int_\epsilon^{\infty} \mathrm{d}r \frac{\partial f(r)}{\partial r} \int \mathrm{d}\Omega \sin^2\theta.
\end{equation}
Integration over the spherical angular coordinates yields:
\begin{equation}
\int \mathrm{d}\Omega \cos^2\theta = \frac{4\pi}{3}, \quad \int \mathrm{d}\Omega \sin^2\theta = \frac{8\pi}{3}.
\end{equation}
The first radial integral, associated with the $\cos^2\theta$ component, is evaluated using integration by parts:
\begin{equation}
\int_\epsilon^{\infty} r \frac{\partial^2 f(r)}{\partial r^2} \,\mathrm{d}r = \left[ r \frac{\partial f(r)}{\partial r} \right]_\epsilon^\infty - \int_\epsilon^{\infty} \frac{\partial f(r)}{\partial r} \,\mathrm{d}r.
\end{equation}
The boundary term at infinity vanishes because the derivatives of a well-behaved test function decay rapidly as $r \to \infty$. In the limit $\epsilon \to 0$, the lower boundary contribution $- \epsilon f'(\epsilon)$ likewise vanishes for any smooth function, leaving only the remaining integral:
\begin{equation}
\lim_{\epsilon\to 0} \left( - \int_\epsilon^{\infty} \frac{\partial f(r)}{\partial r} \,\mathrm{d}r \right) = \lim_{\epsilon\to 0} \Big[ f(\epsilon) - f(\infty) \Big] = f(0).
\end{equation}
The second radial integral, associated with the $\sin^2\theta$ component, is elementary and directly integrates to:
\begin{equation}
\lim_{\epsilon\to 0} \int_\epsilon^{\infty} \frac{\partial f(r)}{\partial r} \,\mathrm{d}r = \lim_{\epsilon\to 0} \Big[ f(\infty) - f(\epsilon) \Big] = -f(0).
\end{equation}
Combining these individual components with their respective angular weights yields:
\begin{equation}
F_{33} = \frac{4\pi}{3} f(0) + \frac{8\pi}{3} \big[ -f(0) \big] = -\frac{4\pi}{3} f(0).
\end{equation}
This establishes equation \maeq{Prop4_Eq}. \hfill $\blacksquare$

\section{Proof of \maprop{prop:laplacian2d} (equation \maeq{Prop6_Eq})}
\label{app:prop6}
By definition, this test integral amounts to
\begin{equation}
\int_{\mathbb{R}^3} \frac{1}{\sqrt{\rho^2+x_3^2}} \, \nabla^2 \Big( h(x_3) g(\rho) \Big) \,\mathrm{d}V.
\end{equation}
Following Ref.~\cite{Halperin} (p.~26), we evaluate the test integral using Green's second identity:
\begin{equation}
\int_{V} \Big( u \nabla^2 v - v \nabla^2 u \Big) \,\mathrm{d}V = \int_{\partial V} \Big( u \vec{\nabla} v - v \vec{\nabla} u \Big) \cdot \mathrm{d}\vec{S},
\label{Green}
\end{equation}
where the integration domain $V$ typically excludes a small sphere around the origin. In the present slab geometry, it is physically more natural to define $V$ by excluding a small \textit{cylinder} ($C$) centered at the origin, characterized by its total height $\epsilon_3$ and radius $\epsilon_\rho$.

We set $u \equiv (\rho^2+x_3^2)^{-1/2}$ and $v \equiv h(x_3) g(\rho)$ in equation \eqref{Green}. Since $\nabla^2 u = 0$ holds identically everywhere within the remaining volume $V$ outside the small cylinder, the volume integral on the left-hand side reduces to:
\begin{equation}
\int_V \frac{1}{\sqrt{\rho^2+x_3^2}} \, \nabla^2 \Big( h(x_3) g(\rho) \Big) \,\mathrm{d}V = \int_{\partial C} \left[ \frac{1}{\sqrt{\rho^2+x_3^2}} \, \vec{\nabla} \big( h(x_3) g(\rho) \big) - \vec{\nabla} \left( \frac{1}{\sqrt{\rho^2+x_3^2}} \right) h(x_3) g(\rho) \right] \cdot \mathrm{d}\vec{S}.
\label{cylinder_surface}
\end{equation}
Note that the surface normal $\mathrm{d}\vec{S}$ points outwards from the domain $V$, which corresponds to an \textit{inward-pointing} normal with respect to the cylinder $C$.

We now estimate the two surface integrals on the right-hand side of equation \eqref{cylinder_surface}. Since the test functions are smooth and well-behaved at the origin, their gradient is bounded there: let $M \equiv \sup_{\partial C} \left| \vec{\nabla}\big(h(x_3)g(\rho)\big) \right| < \infty$. Pulling this bound out leaves the kernel to be integrated over the boundary, which can be done exactly. The two flat caps at $x_3 = \pm\epsilon_3/2$ contribute
\begin{equation}
2 \cdot 2\pi \int_0^{\epsilon_\rho} \frac{\rho \,\mathrm{d}\rho}{\sqrt{\rho^2 + (\epsilon_3/2)^2}} = 4\pi \left( \sqrt{\epsilon_\rho^2 + \epsilon_3^2/4} - \frac{\epsilon_3}{2} \right) \leq 4\pi \epsilon_\rho,
\end{equation}
and the cylindrical mantle at $\rho = \epsilon_\rho$ contributes
\begin{equation}
2\pi \epsilon_\rho \int_{-\epsilon_3/2}^{+\epsilon_3/2} \frac{\mathrm{d}x_3}{\sqrt{\epsilon_\rho^2 + x_3^2}} \leq 2\pi \epsilon_\rho \cdot \frac{\epsilon_3}{\epsilon_\rho} = 2\pi \epsilon_3.
\end{equation}
Together,
\begin{equation}
\left| \int_{\partial C} \frac{1}{\sqrt{\rho^2+x_3^2}} \, \vec{\nabla} \big( h(x_3) g(\rho) \big) \cdot \mathrm{d}\vec{S} \right| \leq 2\pi M \big( 2\epsilon_\rho + \epsilon_3 \big) = \mathcal{O}\big( \epsilon_\rho + \epsilon_3 \big).
\end{equation}
As $(\epsilon_3, \epsilon_\rho) \rightarrow (0,0)$, this first contribution vanishes uniformly---for any aspect ratio $\epsilon_3/\epsilon_\rho$.

For the second surface integral, we exploit the smoothness of the test functions to approximate them by their local boundary values $h(x_3)g(\rho) \approx h(\epsilon_3/2)g(\epsilon_\rho)$. The boundary $\partial C$ splits into two flat caps (at $x_3 = \pm \epsilon_3/2$) and a cylindrical mantle (at $\rho = \epsilon_\rho$).

The combined contribution from the top and bottom caps is given by:
\begin{equation}
- 2 \cdot h(\epsilon_3/2)g(\epsilon_\rho) \cdot 2\pi \int_0^{\epsilon_\rho} \frac{x_3}{(\rho^2+x_3^2)^{3/2}} \, \rho \,\mathrm{d}\rho \;\Bigg\vert_{x_3 =\frac{\epsilon_3}{2}} = - 4\pi \cdot h(\epsilon_3/2)g(\epsilon_\rho) \left( 1 - \frac{\epsilon_3}{\sqrt{\epsilon^2_3 + 4\epsilon^2_\rho}} \right).
\end{equation}
The contribution from the cylindrical mantle evaluates to:
\begin{equation}
- h(\epsilon_3/2)g(\epsilon_\rho) \cdot 2\pi \epsilon_\rho \int_{-\frac{\epsilon_3}{2}}^{+\frac{\epsilon_3}{2}} \frac{\epsilon_\rho}{(\epsilon_\rho^2+x_3^2)^{3/2}} \,\mathrm{d}x_3 = - 4\pi \cdot h(\epsilon_3/2)g(\epsilon_\rho) \left( \frac{\epsilon_3}{\sqrt{\epsilon^2_3 + 4\epsilon^2_\rho}} \right).
\end{equation}
Summing the contributions from the caps and the mantle, the terms containing the geometric parameters $\epsilon_3$ and $\epsilon_\rho$ cancel out exactly. The total boundary integral yields $-4\pi h(0) g(0)$ in the limit, completely independent of the specific geometric aspect ratio chosen to approach the origin.
\hfill $\blacksquare$

\section{A conjecture about $F_\parallel$ in the limit $d/\Lambda \rightarrow 0$}
\label{app:slab}

\subsection{A simple step-like test function}
We evaluate $F_\parallel$ for a class of idealized, step-like test functions. In accordance with the symmetry $h(x_3) = h(-x_3)$ assumed for the slab model, the confinement profile is centred on the dipole site: $h(x_3) = h_0$ for $|x_3| \leq d/2$ and $h(x_3) = 0$ otherwise, so that the slab of thickness $d$ straddles the plane $x_3 = 0$ rather than resting on it. The in-plane profile is $g(\rho) = g_0$ for $\rho \in [0,\Lambda]$ and $g(\rho) = 0$ for $\rho > \Lambda$.

It is convenient to perform the $x_3$-integration in equation \maeq{F_parallel_def} first. We therefore introduce the kernel
\begin{equation}
K(\rho) \equiv \int \mathrm{d}x_3 \, \frac{h(x_3)}{\sqrt{\rho^2 + x_3^2}} = h_0 \int_{-d/2}^{+d/2} \frac{\mathrm{d}x_3}{\sqrt{\rho^2 + x_3^2}} = 2 h_0 \operatorname{arcsinh}\left( \frac{d}{2\rho} \right),
\label{Kernel}
\end{equation}
which is not to be confused with the spherically symmetric test function $f(|\vec{r}|)$ of the isotropic case. In terms of $K$, equation \maeq{F_parallel_def} reads
\begin{equation}
F_\parallel = \pi \int_0^\infty \rho \, K(\rho) \, \nabla^2_\rho g(\rho) \,\mathrm{d}\rho,
\end{equation}
the prefactor $\pi$ collecting the azimuthal integration together with the fact that $\nabla^2_\rho$ delivers $F_{11} + F_{22} = 2F_\parallel$ rather than $F_\parallel$ itself.

To evaluate this integral, we employ the radial representation of the two-dimensional Laplacian, $\nabla_\rho^2 = \frac{1}{\rho} \frac{\partial}{\partial \rho} \left( \rho \frac{\partial}{\partial \rho} \right)$. Substituting this operator into the integrand cancels the explicit prefactor of $\rho$:
\begin{equation}
F_\parallel = \pi \int_0^\Lambda K(\rho) \frac{\partial}{\partial \rho} \left( \rho \frac{\partial g(\rho)}{\partial \rho} \right) \mathrm{d}\rho.
\end{equation}
Since $g(\rho)$ is a sharp step function dropping to zero at $\rho = \Lambda$, its derivative contains a Dirac delta function. To capture this boundary contribution, we extend the integration limit to $\Lambda^+$. Integrating by parts yields:
\begin{equation}
F_\parallel = \pi \left[ K(\rho) \cdot \rho \frac{\partial g(\rho)}{\partial \rho} \right]_0^{\Lambda^+} - \pi \int_0^{\Lambda^+} \frac{\partial K(\rho)}{\partial \rho} \left( \rho \frac{\partial g(\rho)}{\partial \rho} \right) \mathrm{d}\rho.
\end{equation}
The boundary term vanishes at the lower limit $\rho = 0$ due to the factor of $\rho$---the kernel \eqref{Kernel} diverges there only logarithmically, so that $\rho K(\rho) \to 0$---and at the upper limit $\rho = \Lambda^+$ because $g(\rho)$ becomes a constant zero outside the disk. Expressing the step-like profile as $g(\rho) = g_0 \Theta(\Lambda - \rho)$, its distributional derivative reads $\frac{\partial g(\rho)}{\partial \rho} = -g_0 \delta(\rho - \Lambda)$. Substituting this distribution into the remaining integral reduces it to a boundary evaluation:
\begin{equation}
F_\parallel = - \pi \int_0^{\Lambda^+} \frac{\partial K(\rho)}{\partial \rho} \cdot \rho \cdot \Big[-g_0 \delta(\rho - \Lambda)\Big] \,\mathrm{d}\rho = \pi g_0 \Lambda \left. \frac{\partial K(\rho)}{\partial \rho} \right\vert_{\rho = \Lambda}.
\end{equation}
We evaluate the derivative $\frac{\partial K(\rho)}{\partial \rho}$ by differentiating under the integral sign:
\begin{equation}
\frac{\partial K(\rho)}{\partial \rho} = -h_0 \rho \int_{-d/2}^{+d/2} \frac{\mathrm{d}x_3}{(\rho^2 + x_3^2)^{3/2}} = -h_0 \rho \left[ \frac{x_3}{\rho^2 \sqrt{\rho^2 + x_3^2}} \right]_{-d/2}^{+d/2} = -\frac{h_0 \, d}{\rho \sqrt{\rho^2 + d^2/4}}.
\end{equation}
Evaluating this derivative at the boundary $\rho = \Lambda$ yields:
\begin{equation}
\left. \frac{\partial K(\rho)}{\partial \rho} \right\vert_{\rho = \Lambda} = -\frac{h_0 \, d}{\Lambda \sqrt{\Lambda^2 + d^2/4}}.
\end{equation}
Multiplying by the scaling factor $\pi g_0 \Lambda$ gives the exact closed-form value for this model profile:
\begin{equation}
F_\parallel = -\frac{\pi h_0 g_0 \, d}{\sqrt{\Lambda^2 + d^2/4}} = -\frac{2 \pi h_0 g_0 \, d}{\sqrt{4\Lambda^2 + d^2}}.
\end{equation}
For $d/\Lambda \ll 1$ we have $F_\parallel = \mathcal{O}(d/\Lambda)$, i.e.\ for this class of test functions $F_\parallel$ vanishes in the limit $d/\Lambda \rightarrow 0$.

\subsection{A step-like profile for $h(x_3)$ with an exponential distribution $g(\rho) = g_0 \mathrm{e}^{-\rho/\Lambda}$}
We retain the symmetric slab profile $h(x_3)$ of the previous subsection, and hence the kernel \eqref{Kernel}. Employing the radial form of the Laplace operator as before:
\begin{equation}
F_\parallel = \pi \int_0^\infty K(\rho) \frac{\partial}{\partial \rho} \left( \rho \frac{\partial g(\rho)}{\partial \rho} \right) \mathrm{d}\rho.
\end{equation}
Because the surface boundary terms vanish identically for an exponentially decaying profile, integration by parts gives:
\begin{equation}
F_\parallel = -\pi \int_0^\infty \frac{\partial K(\rho)}{\partial \rho} \cdot \rho \frac{\partial g(\rho)}{\partial \rho} \,\mathrm{d}\rho = -\pi h_0 g_0 \frac{d}{\Lambda} \int_0^\infty \frac{\mathrm{e}^{-\rho/\Lambda}}{\sqrt{\rho^2 + d^2/4}} \,\mathrm{d}\rho.
\end{equation}
To transform this expression into a standard integral identity, we apply the substitution $\rho = \Lambda t$ (with $\mathrm{d}\rho = \Lambda \,\mathrm{d}t$), which factors out the structural scales:
\begin{equation}
F_\parallel = -\pi h_0 g_0 \frac{d}{\Lambda} \int_0^\infty \frac{\mathrm{e}^{-t}}{\sqrt{t^2 + \left(\frac{d}{2\Lambda}\right)^2}} \,\mathrm{d}t.
\end{equation}
This definite integral can be evaluated exactly in terms of the Struve function of the first kind $\mathbf{H}_0$ and the Neumann function (Bessel function of the second kind) $Y_0$ using the known integral identity:
\begin{equation}
\int_0^\infty \frac{\mathrm{e}^{-t}}{\sqrt{t^2 + z^2}} \,\mathrm{d}t = \frac{\pi}{2} \left[ \mathbf{H}_0(z) - Y_0(z) \right].
\end{equation}
Setting the dimensionless scaling argument to $z = d/2\Lambda$---half the value it would take for a slab resting on the plane $x_3 = 0$---the closed-form expression for $F_\parallel$ reads:
\begin{equation}
F_\parallel = -h_0 g_0 \frac{\pi^2 d}{2\Lambda} \left[ \mathbf{H}_0\left(\frac{d}{2\Lambda}\right) - Y_0\left(\frac{d}{2\Lambda}\right) \right].
\end{equation}
We examine the asymptotic behavior of this expression in the tight-confinement limit ($d/\Lambda \to 0$). Utilizing the power-series expansion of the Struve function,
\begin{equation}
\mathbf{H}_0(z) = \frac{2}{\pi} z + \mathcal{O}(z^3),
\end{equation}
and the small-argument asymptotic expansion of the Neumann function,
\begin{equation}
Y_0(z) = \frac{2}{\pi} \left[ \ln\left(\frac{z}{2}\right) + \gamma \right] + \mathcal{O}(z^2),
\end{equation}
where $\gamma \approx 0.5772$ is the Euler--Mascheroni constant, we substitute these expressions back into the total formulation. Retaining terms up to first order in $z$, the asymptotic expansion of $F_\parallel$ yields:
\begin{equation}
F_\parallel \approx \pi h_0 g_0 \frac{d}{\Lambda} \ln\left(\frac{d}{\Lambda}\right) - \pi h_0 g_0 \frac{d}{\Lambda} \left[ \ln(4) - \gamma \right] + \mathcal{O}\left(\left(\frac{d}{\Lambda}\right)^2\right).
\label{AppE_Asymptotic}
\end{equation}
This confirms that $F_\parallel$ vanishes in the pure two-dimensional limit ($d/\Lambda \to 0$), scaling predominantly as $\mathcal{O}\big(\frac{d}{\Lambda} \ln \frac{d}{\Lambda}\big)$.

\section{Proof of equation \maeq{Magneto} (\maprop{prop:magnetostatic})}
\label{app:magneto}
The magnetic field of a continuous magnetization distribution $\vec{M}(\vec{r})$ is given by:
\begin{equation}
\vec{B}(\vec{r}) = \frac{\mu_0}{4\pi} \vec{\nabla} \times \int \mathrm{d}V' \, \big(\vec{\nabla}' \times \vec{M}(\vec{r}')\big) G(|\vec{r} - \vec{r}'|),
\end{equation}
where $G(|\vec{r} - \vec{r}'|) \equiv |\vec{r} - \vec{r}'|^{-1}$. In the following steps, we abbreviate the Green's function simply as $G$. We employ the product rule of vector analysis,
\begin{equation}
\vec{\nabla}' \times (G \vec{M}') = G (\vec{\nabla}' \times \vec{M}') + (\vec{\nabla}' G) \times \vec{M}' = G (\vec{\nabla}' \times \vec{M}') - \vec{M}' \times \vec{\nabla}' G,
\end{equation}
to isolate the integrand term:
\begin{equation}
G (\vec{\nabla}' \times \vec{M}') = \vec{\nabla}' \times (G \vec{M}') + \vec{M}' \times \vec{\nabla}' G.
\end{equation}
Integrating this expression over the entire source volume $V'$ yields:
\begin{equation}
\int \mathrm{d}V' \, G (\vec{\nabla}' \times \vec{M}') = \int \mathrm{d}V' \, \vec{\nabla}' \times (G \vec{M}') + \int \mathrm{d}V' \, \vec{M}' \times \vec{\nabla}' G.
\end{equation}
The first integral on the right-hand side can be transformed into a surface integral via the divergence theorem variant for curls:
\begin{equation}
\int \mathrm{d}V' \, \vec{\nabla}' \times (G \vec{M}') = \oint_{\partial V'} \mathrm{d}\vec{S}' \times (G \vec{M}').
\end{equation}
Since the magnetization field $\vec{M}(\vec{r}')$ is bounded and localized, the bounding surface $\partial V'$ can be pushed to infinity, causing this surface integral to vanish identically. Utilizing the relation $\vec{\nabla}' G = -\vec{\nabla} G$, we obtain:
\begin{equation}
\int \mathrm{d}V' \, G (\vec{\nabla}' \times \vec{M}') = \int \mathrm{d}V' \, \vec{M}' \times \vec{\nabla}' G = \int \mathrm{d}V' \, (\vec{\nabla} G) \times \vec{M}'.
\end{equation}
Because the unprimed operator $\vec{\nabla}$ acts exclusively on the observation coordinate $\vec{r}$, we can rewrite the integrand as $(\vec{\nabla} G) \times \vec{M}' = \vec{\nabla} \times (G \vec{M}')$. Factoring the curl out of the volume integration leads to:
\begin{equation}
\vec{B}(\vec{r}) = \frac{\mu_0}{4\pi}\,\vec{\nabla} \times \left( \vec{\nabla} \times \int \mathrm{d}V' \frac{\vec{M}(\vec{r}')}{|\vec{r}-\vec{r}'|} \right).
\end{equation}
Applying the standard double-curl vector identity $\vec{\nabla} \times (\vec{\nabla} \times \vec{X}) = \vec{\nabla}(\vec{\nabla} \cdot \vec{X}) - \nabla^2 \vec{X}$, the magnetic field splits into two components:
\begin{equation}
\vec{B}(\vec{r}) = \frac{\mu_0}{4\pi} \vec{\nabla} \left( \vec{\nabla} \cdot \int \mathrm{d}V' \frac{\vec{M}(\vec{r}')}{|\vec{r}-\vec{r}'|} \right) - \frac{\mu_0}{4\pi} \int \mathrm{d}V' \, \vec{M}(\vec{r}') \nabla^2 \left( \frac{1}{|\vec{r}-\vec{r}'|} \right).
\end{equation}
We now pull the unprimed differential operators inside the first integral, where they act on the kernel $|\vec{r}-\vec{r}'|^{-1}$ alone. The step that makes the result recognizable is that this kernel depends on $\vec{r}$ and $\vec{r}'$ only through their difference, so that each unprimed derivative may be traded for a primed one at the cost of a sign,
\begin{equation}
\partial_i \frac{1}{|\vec{r}-\vec{r}'|} = -\partial'_i \frac{1}{|\vec{r}-\vec{r}'|}
\quad \Longrightarrow \quad
\partial_i \partial_j \frac{1}{|\vec{r}-\vec{r}'|} = \partial'_i \partial'_j \frac{1}{|\vec{r}-\vec{r}'|},
\end{equation}
the two sign changes cancelling for the second derivatives that appear here. Setting $\vec{r} = \vec{0}$ and renaming the integration variable $\vec{r}' \rightarrow \vec{r}$ therefore converts the observation-point derivatives into exactly the source-point derivatives of the test integral \maeq{Testintegral}. Substituting the specialized magnetization profile $\vec{M}(\vec{r}) = \vec{\mu} f(\vec{r})$ and noting that $\nabla^2 |\vec{r}-\vec{r}'|^{-1} = -4\pi \delta(\vec{r}-\vec{r}')$, the second integral evaluates exactly via the Dirac delta distribution. Evaluating the total expression at the origin ($\vec{r} = \vec{0}$) yields:
\begin{equation}
\vec{B}(\vec{0}) = \frac{\mu_0}{4\pi} \int \mathrm{d}V f(\vec{r}) \vec{\nabla} \left( \vec{\nabla} \cdot \frac{\vec{\mu}}{|\vec{r}|} \right) + \mu_0 \vec{\mu} f(\vec{0}).
\end{equation}
This macroscopic expression matches the distribution-theoretic test integral over the field of a single localized point dipole, completing the proof. \hfill $\blacksquare$


\begin{thebibliography}{99}

\bibitem{Maxwell}
Maxwell J C 1873 \textit{A Treatise on Electricity and Magnetism} vol~2 (Oxford: Clarendon Press) pp~418--419

\bibitem{Fermi}
Fermi E 1930 \textit{Z. Phys.} \textbf{60} 320--333 \\
\url{https://link.springer.com/article/10.1007/BF01339933}

\bibitem{Jackson_40}
Jackson J D 1977 \textit{The nature of intrinsic magnetic dipole moments: what can the famous 21 cm astrophysical spectral line of atomic hydrogen tell us about the nature of magnetic dipoles?} \textit{CERN Yellow Reports: Monographs} \\
\url{https://doi.org/10.5170/CERN-1977-017}

\bibitem{GriffithsAJP}
Griffiths D J 1982 \textit{Am. J. Phys.} \textbf{50} 698--703

\bibitem{Jackson}
Jackson J D 1999 \textit{Classical Electrodynamics} 3rd edn (New York: John Wiley and Sons)

\bibitem{Griffiths}
Griffiths D J 2012 \textit{Introduction to Electrodynamics} 4th edn (Englewood Cliffs, NJ: Prentice-Hall)

\bibitem{Frahm}
Frahm C P 1983 Some novel delta function identities \textit{Am. J. Phys.} \textbf{51} 826--829 \\
\url{https://doi.org/10.1119/1.13127}

\bibitem{Muniz}
Mu\~niz Y, Fonseca A and Farina C 2019 Electrostatic and magnetostatic fields of point dipoles revisited \textit{Rev. Mex. F\'is. E} \textbf{65} 71--76

\bibitem{NIST}
Olver F W J, Olde Daalhuis A B, Lozier D W, Schneider B I, Boisvert R F, Clark C W, Miller B R, Saunders B V, Cohl H S and McClain M A (eds) 2026 \textit{NIST Digital Library of Mathematical Functions} Release 1.2.7 of 2026-06-15 \\
\url{https://dlmf.nist.gov/1.16}

\bibitem{Light}
Lighthill M J 1959 \textit{An Introduction to Fourier Analysis and Generalised Functions} (Cambridge: Cambridge University Press)

\bibitem{Halperin}
Halperin I 1952 \textit{Introduction to the Theory of Distributions} Canadian Mathematical Congress, Lecture Series, No. 1 (Toronto: University of Toronto Press) p~26

\bibitem{Wigner}
Wigner E 1927 \textit{Z. Phys.} \textbf{43} 624--652 \\
\url{https://doi.org/10.1007/BF01397327}; \\
van der Waerden B L 1974 \textit{Group Theory and Quantum Mechanics} (Berlin: Springer) \\
\url{https://doi.org/10.1007/978-3-642-65860-0}

\bibitem{DP}
Pescia D 2024 \textit{Lectures on Symmetry-Assisted Computation} (Singapore: World Scientific)

\bibitem{Yang}
Yang X L, Guo S H, Chan F T, Wong K W and Ching W Y 1991 \textit{Phys. Rev.} A \textbf{43} 1186--1196

\bibitem{Xie}
Xie W 2010 \textit{Nucl. Instrum. Methods Phys. Res.} B \textbf{268} 3321--3324

\bibitem{Bastard}
Bastard G 1981 \textit{Phys. Rev.} B \textbf{24} 4714--4722 \\
\url{https://journals.aps.org/prb/pdf/10.1103/PhysRevB.24.4714}; \\
Bastard G 1988 \textit{Wave Mechanics Applied to Semiconductor Heterostructures} (Les Ulis: Les \'Editions de Physique) ch~4

\bibitem{PRL}
Perraud S, Kanisawa K, Wang Z Z and Fujisawa T 2008 \textit{Phys. Rev. Lett.} \textbf{100} 056806 \\
\url{https://doi.org/10.1103/PhysRevLett.100.056806}

\bibitem{Plat}
Lunde A M and Platero G 2013 \textit{Phys. Rev.} B \textbf{88} 115411, and references therein \\
\url{https://doi.org/10.1103/PhysRevB.88.115411}

\bibitem{Coi}
Coish W A and Baugh J 2009 \textit{Phys. Status Solidi} B \textbf{246} 2203--2215 \\
\url{https://onlinelibrary.wiley.com/doi/epdf/10.1002/pssb.200945229}

\bibitem{Mer}
Merkulov I A, Efros Al L and Rosen M 2002 \textit{Phys. Rev.} B \textbf{65} 205309 \\
\url{https://doi.org/10.1103/PhysRevB.65.205309}


\bibitem{PIP}
Pescia D and Vindigni A 2024 On the spontaneous magnetization of two-dimensional ferromagnets \textit{Papers in Physics} \textbf{16} 160001 \\
\url{https://doi.org/10.4279/PIP.160001}

\bibitem{Yves}
Buess M, Acremann Y, Kashuba A, Back C H and Pescia D 2003 \textit{J. Phys.: Condens. Matter} \textbf{15} R1093--R1100


\end{thebibliography}
\end{document}